\RequirePackage{fix-cm}

\documentclass[smallextended]{svjour3}       

\smartqed  

\usepackage{amsfonts,amssymb,amsmath}
\usepackage{graphics,graphicx,epsfig}
\usepackage[dvipsnames]{xcolor}
\usepackage[normalem]{ulem}
\usepackage{comment}
\usepackage{epstopdf}
\usepackage{subfigure}
\usepackage{marvosym}
\usepackage{hyperref}
\usepackage{geometry}
\usepackage{cite}
 \newcommand{\be}{\begin{eqnarray}}
 \newcommand{\ee}{\end{eqnarray}}

\begin{document}

\title{Revealing hidden standard tripartite nonlocality by local filtering}

\author{Qiao-Qiao Lv$^1$  \and Jin-Min Liang $^1$  \and Zhi-Xi Wang$^1$  \and Shao-Ming Fei$^{1,2}$   }
\institute{ Zhi-Xi Wang\at
              \email{wangzhx@cnu.edu.cn}\\
           \and Shao-Ming Fei \at
             \email{feishm@cnu.edu.cn}\\
 \at 1 School of Mathematical Sciences,  Capital Normal University,  Beijing 100048,  China\\
 \at 2  Max Planck Institute for Mathematics in the Sciences, Leipzig 04103, Germany
}



\maketitle

\begin{abstract}
\label{intro}
Quantum nonlocality is a kind of significant quantum correlation that is stronger than quantum entanglement and EPR steering. The standard tripartite nonlocality can be detected by the violation of the Mermin inequality. By using local filtering operations, we give a tight upper bound on the maximal expected value of the Mermin operators. By detailed examples we show that the hidden standard nonlocality can be revealed by local filtering which can enhance the robustness of the noised entangled states.

\end{abstract}

\section{Introduction}
Nonlocal quantum correlations are the essential nature of quantum physics \cite{EPRnonlocal}. For bipartite systems the Bell nonlocality \cite{BellNonlocal,Wiseman,StrongNonlocal} gives rise to stronger quantum correlations than quantum entanglement \cite{HHHH2009entangle,LS2019,CG2005} and EPR steering \cite{steer2020}. As a fascinating counterintuitive phenomenon related to the foundation of quantum mechanics \cite{Bell1966}, nonlocality is understood as a resource in quantum information process ranging form quantum computation \cite{Liang2020,Liang2022}, quantum key distributed \cite{AK1991} and random numbers certification \cite{Nature2010}. It is significant to detect the nonlocality of given quantum systems.

A bipartite quantum state $\rho_{AB}$ is said to be locally correlated if the joint probability distributions satisfies \cite{Wiseman},
\begin{eqnarray}\nonumber
P(a, b|A_x, B_y)=\sum_{\lambda}p_{\lambda}P(a|A_x,\lambda)P(b|B_y,\lambda),
\end{eqnarray}
where $p_{\lambda}$ is the probability distribution over the hidden variables $\lambda$, $p_{\lambda}\geq 0$, $\sum_{\lambda}p_{\lambda}=1$, $P(a,b|A_x, B_y)$ denotes the joint probability distribution that Alice performs measurement $A_x$ on subsystem $A$ with outcome $a$, and Bob performs measurement $B_y$ on subsystem $B$ with outcome $b$, $P(a|A_x, \lambda)$ denotes the conditional probability of getting outcome $a$ when Alice performs the measurement $A_x$ on her subsystem, $P(b|B_y, \lambda)$ is similarly defined. The bipartite nonlocality can be detected by the violation of a kind of Bell inequality, the CHSH inequality \cite{CHSH1969}.

The maximal violation of a Bell inequality can be enhanced by local filtering operations \cite{CHSHviolate}. In \cite{filterCHSH} the authors presented a class of two-qubit entangled states admitting local hidden variable models, and shew that these states violate a Bell inequality after the local filtering. Namely, there exist entangled states whose so called hidden non-locality can be revealed by using a sequence of measurements. In fact, local filtering operations can not only reveal hidden quantum nonlocality, but also hidden quantum steerability \cite{FilterSteer}.

For tripartite systems, the Mermin inequality \cite{Mermin1990} is a natural generalization of the CHSH inequality, which can be violated by not only genuine tripartite nonlocal states but also by standard tripartite nonlocal states. The upper bound on the maximal expectation value of the Mermin operator has been nicely derived in \cite{QIP2019}, though not always tight.
Similar to the enhanced maximal violation of Bell inequalities, it is interesting to study the
enhancement of the maximal violation of the Mermin inequality under local filtering.

In this work, we investigate the maximal violation of the Mermin inequality under local filtering. We first analyze and obtain the tight upper bounds on the maximal expected value of the Mermin operator under local filtering. Then applying our results to some special quantum states, the isotropic states and the noisy GHZ states, we show that local filtering can reveal the hidden standard nonlocality. Moreover, it is shown that the local filtering can transform the initial noisy state to a state with stronger tripartite nonlocality.

\section{Tight upper bound on Mermin operator under local filtering}
\label{sec:1}
A tripartite state $\rho_{ABC}$ is fully locally correlated if the joint probability distribution admits a local hidden variable (LHV) model \cite{Wiseman}, namely,
\begin{eqnarray}\nonumber
\begin{aligned}
P(a,b,c|A_x, B_y, C_z)=\sum_{\lambda}p_{\lambda}P(a|A_x, \lambda)P(b|B_y, \lambda)P(c|C_z, \lambda)
\end{aligned}
\end{eqnarray}
for all $x, y, z, a, b$ and $c$, where $P(a,b,c|A_x, B_y, C_z)$ is the joint probability when Alice, Bob and Charlie perform local measurements $A_x$, $B_y$ and $C_z$ with outcomes $a$, $b$ and $c$, respectively, $p_{\lambda}\geq 0$ is the probability distribution over the hidden variable $\lambda$, $\sum_{\lambda}p_{\lambda}=1$, $P(a|A_x, \lambda)$ denotes the conditional probability of obtaining outcome $a$ when Alice performs the measurement $A_x$ on her subsystem, $P(b|B_y, \lambda)$ and $P(c|C_z, \lambda)$ are similarly defined.

The tripartite non-locality of arbitrary 3-qubit states $\rho$ can be detected by the violation of the Mermin inequality $|\langle\mathcal{M}\rangle_{\rho}|\equiv |\textrm{Tr}[\mathcal{M} \rho]|\leq 2$. The Mermin operator $\mathcal{M}$ has the form \cite{Mermin1990},
$$\mathcal{M}=A_0\otimes (B_0C_1+B_1C_0)+A_1\otimes (B_0C_0-B_1C_1),\nonumber$$
where $A_0$, $A_1$, $B_0$, $B_1$, $C_0$ and $C_1$ are quantum mechanical observables of the form $K=\vec{k}\cdot\vec{\sigma}=\sum\limits ^{3}_{i=1}k_i\sigma_i$, with a unit vector $\vec{k}\in\{\vec{a}, \vec{a}', \vec{b}, \vec{b}', \vec{c}, \vec{c}'\}$, the standard Pauli matrices  $\sigma_i$, $K\in\{A_0, A_1, B_0, B_1, C_0, C_1\}$.

In a recent work \cite{QIP2019} it has been shown that the maximal expectation value of the Mermin operator for arbitrary 3-qubit state $\rho$ is given by $\mathcal{Q}_{\mathcal{M}}:= \max|\langle\mathcal{M}\rangle_{\rho}|\leq 2\sqrt{2}\lambda_{\max}$, where $\lambda_{\max}$ is the largest singular value of matrix $C$, $C=(C_{j,ik})$ is the correlation matrix of $\rho$ with entries given by $C_{ijk}=\textrm{Tr}[\rho(\sigma_i\otimes\sigma_j\otimes\sigma_k)]$, $i, j, k=1, 2, 3$. This upper bound is tight if the degeneracy with respect to the largest singular value $\lambda_{\max}$ is more than 1, and the two degenerate nine-dimensional singular vectors corresponding to $\lambda_{\max}$ take the forms of $\vec{a}\otimes\vec{c}-\vec{a}'\otimes\vec{c}'$ and $\vec{a}\otimes\vec{c}'+\vec{a}'\otimes\vec{c}$.

Next, we investigate the violations of the Mermin inequality under local filtering, by
computing the maximal expectation values of the Mermin operators with respect to the locally filtered 3-qubit states. For any 3-qubit state $\rho$, after local filtering one gets \cite{filterCHSH},
\begin{eqnarray}\label{rhop}
\begin{aligned}
\rho'=\frac{1}{F}(F_A\otimes F_B\otimes F_C)\rho(F_A\otimes F_B\otimes F_C)^{\dag},
\end{aligned}
\end{eqnarray}
where $F_A$, $F_B$ and $F_C$ are positive operators acted locally on the three subsystems respectively, $F=\textrm{Tr}[(F_A\otimes F_B\otimes F_C)\rho(F_A\otimes F_B\otimes F_C)^{\dag}]$ is the normalization constant. Suppose the filter operators $F_A$, $F_B$ and $F_C$ have the following spectral decompositions,
\begin{align}\label{filter}
F_A=U\Sigma_A U^{\dag},~~ F_B=V\Sigma_B V^{\dag},~~ F_C=W\Sigma_C W^{\dag},
\end{align}
where $U$, $V$ and $W$ are unitary operators. Set
\begin{eqnarray}
\begin{aligned}
\alpha_i=\Sigma_A\sigma_i\Sigma_A,~~ \beta_j=\Sigma_B\sigma_j\Sigma_B,~~ \gamma_k=\Sigma_C\sigma_k\Sigma_C,
\end{aligned}
\end{eqnarray}
for $i, j, k=1, 2, 3$. Without loss of generality, we assume that the singular matrices
have the forms,
\begin{eqnarray}
\begin{aligned}
\Sigma_A=\left(\begin{array}{cc}
  l&0\\
  0&1
\end{array}\right),~~
\Sigma_B=\left(\begin{array}{cc}
  m&0\\
  0&1
\end{array}\right),~~
\Sigma_C=\left(\begin{array}{cc}
  n&0\\
  0&1
\end{array}\right),
\end{aligned}
\end{eqnarray}
with $l,m,n\geq 0$.

We have the following theorem which provides a tight upper bound on the maximal violation value of the Mermin inequality.

\begin{theorem}
For an arbitrary 3-qubit quantum state $\rho$, the maximal expectation value of the Mermin operator for the filtered state $\rho'$ satifies
\begin{eqnarray}\label{QM}
\begin{aligned}
\max|\langle\mathcal{M}\rangle_{\rho'}|\leq 2\sqrt{2}\lambda{'}_{\max},
\end{aligned}
\end{eqnarray}
where $\lambda{'}_{\max}$ is the maximal singular value of $\frac{\tilde{D}}{F}$, where $\tilde{D}=(\tilde{D}_{j,ik})$ with $\tilde{D}_{ijk}=\textrm{Tr}[\tilde{\rho}(\alpha_i\otimes \beta_j\otimes \gamma_k)]$, $\tilde{\rho}$ is a state that is locally unitary equivalent to $\rho$.
\end{theorem}

\begin{proof}
Based on the dual relation between $SU(2)$ and $SO(3)$ \cite{SO1995}, we have $U\sigma_i U^{\dag}=\sum_{i{'}}O_{ii{'}}\sigma_{i{'}}$, where $U$ is a unitary operator and $O=(O_{ij})$ belongs to $SO(3)$. Therefore, we have
\begin{eqnarray}
\begin{aligned}
C'_{ijk}&=\textrm{Tr}[\rho'(\sigma_i\otimes\sigma_j\otimes\sigma_k)]\\
&=\textrm{Tr}[\frac{(F_A\otimes F_B\otimes F_C)\rho(F_A\otimes F_B\otimes F_C)^{\dag}}{F}(\sigma_i\otimes\sigma_j\otimes\sigma_k)]\\
&=\frac{1}{F}\textrm{Tr}[\rho(U\Sigma_A U^{\dag}\otimes V\Sigma_B V^{\dag}\otimes W\Sigma_C W^{\dag})(\sigma_i\otimes\sigma_j\otimes\sigma_k)(U\Sigma_A U^{\dag}\otimes V\Sigma_B V^{\dag}\otimes W\Sigma_C W^{\dag})]\\
&=\frac{1}{F}\textrm{Tr}[(U^{\dag}\otimes V^{\dag}\otimes W^{\dag})\rho(U\otimes V\otimes W)(\Sigma_AU^{\dag}\sigma_i U\Sigma_A\otimes\Sigma_BV^{\dag}\sigma_j V\Sigma_B\otimes\Sigma_CW^{\dag}\sigma_k W\Sigma_C)]\\
&=\frac{1}{F}\textrm{Tr}[\tilde{\rho}(\Sigma_A\sum\limits_{i'}O_{ii'}\sigma_{i'}\Sigma_A\otimes
\Sigma_B\sum\limits_{j'}O_{jj'}\sigma_{j'}\Sigma_B\otimes\Sigma_C\sum\limits_{k'}O_{kk'}\sigma_{k'}\Sigma_C)]\\
&=\frac{1}{F}\sum\limits_{i'j'k'}O_{ii'}O_{jj'}O_{kk'}\textrm{Tr}[\tilde{\rho}(\Sigma_A\sigma_{i'}\Sigma_A\otimes
\Sigma_B\sigma_{j'}\Sigma_B\otimes\Sigma_C\sigma_{k'}\Sigma_C)]\\
&=\frac{1}{F}\sum\limits_{i'j'k'}O_{ii'}O_{jj'}O_{kk'}\textrm{Tr}[\tilde{\rho}(\alpha_{i'}\otimes \beta_{j'}\otimes \gamma_{k'})]\\
&=\frac{1}{F}[O_A \tilde{D} (O^{T}_{B}\otimes O^{T}_{C})]_{ijk}.
\end{aligned}
\end{eqnarray}
Hence, $C'=\frac{O_A \tilde{D} (O^{T}_{B}\otimes O^{T}_{C})}{F}$, and $(C')^{\dag}C'=\frac{1}{F^2}(O_B\otimes O_C)\tilde{D}^{\dag}O_A^{\dag}O_A \tilde{D}(O_B\otimes O_C)^{\dag}=\frac{1}{F^2}(O_B\otimes O_C)\tilde{D}^{\dag} \tilde{D}(O_B\otimes O_C)^{\dag}$. As $O_B$ and $O_C$ belong to $SO(3)$, $(C')^{\dag}C'$ has the same eigenvalues as $\frac{\tilde{D}^{\dag}\tilde{D}}{F^2}$. That is to say, $\lambda{'}_{\max}$ is also the maximal singular value of $\frac{\tilde{D}}{F}$.
\end{proof}

\textit{\textbf{Remark }} The normalization factor $F$ has the following form
\begin{eqnarray}\nonumber
\begin{aligned}
F&=\textrm{Tr}[(F_A\otimes F_B\otimes F_C)\rho(F_A\otimes F_B\otimes F_C)^{\dag}]\\
&=\textrm{Tr}[(U\Sigma_A U^{\dag}\otimes V\Sigma_B V^{\dag}\otimes W\Sigma_C W^{\dag})\rho(U\Sigma_A U^{\dag}\otimes V\Sigma_B V^{\dag}\otimes W\Sigma_C W^{\dag})]\\
&=\textrm{Tr}[\rho(U\Sigma^2_A U^{\dag}\otimes V\Sigma^2_B V^{\dag}\otimes W\Sigma^2_C W^{\dag})]\\
&=\textrm{Tr}[(U^{\dag}\otimes V^{\dag}\otimes W^{\dag})\rho(U\otimes V\otimes W)(\Sigma^2_A\otimes\Sigma^2_B\otimes\Sigma^2_C)]\\
&=\textrm{Tr}[\tilde{\rho}(\Sigma^2_A\otimes\Sigma^2_B\otimes\Sigma^2_C)],
\end{aligned}
\end{eqnarray}
where $\tilde{\rho}$ and $\rho$ are local unitary equivalent. They have the same maximal violation value of the Mermin inequality.

Note that the inequality (\ref{QM}) saturates if the degeneracy of $\lambda'_{\max}$ is more than 1, and the two nine-dimensional singular vectors corresponding to $\lambda_{\max}$ take the forms of $\vec{a}\otimes\vec{c}-\vec{a}'\otimes\vec{c}'$ and  $\vec{a}\otimes\vec{c}{'}+\vec{a}{'}\otimes\vec{c}$, respectively.

To illustrate the theorem let us consider the following examples.

\textit{\textbf{Example 1.}} Consider the 3-qubit mixed Greenberger-Horne-Zeilinger (GHZ) state \cite{EX12015},
$$\rho_{GHZ}=p|GHZ\rangle\langle GHZ|+\frac{1-p}{4} I_2\otimes\tilde{I},$$
where $0\leq p\leq 1$, $|GHZ\rangle=\frac{|000\rangle+|111\rangle}{\sqrt{2}}$, $I_2$ is the $2\times 2$ identity matrix and $\tilde{I}= $diag$(1, 0, 0, 1)$.
The state $\rho_{GHZ}$ is shown to be genuine multipartite entangled for $\frac{1}{3}< p\leq 1$, and it admits bilocal hidden model for $0\leq p\leq 0.41667$ \cite{EX12015}. Later, Li \emph{et al.} pointed out that $\rho_{GHZ}$ is genuine multipartite nonlocal \cite{GN2013,LM2017} for $0.707107<p\leq 1$, namely, it violates the Svetlichny inequality (SI) \cite{Svetlichny1987} when $0.707107<p\leq 1$. The maximal violation of the Svetlichny inequality under local filtering has been also calculated \cite{FOP}. Recently, the upper bound of the Mermin operator has been studied in \cite{QIP2019}, which shows that the state violates the Mermin inequality if $\frac{1}{2}<p\leq 1$, i.e., the state is standard nonlocal for $\frac{1}{2}<p\leq 1$.

By direct calculation, we have the correlation matrix of $\rho_{GHZ}$,
\begin{eqnarray}
\begin{aligned}
C=\left(\begin{array}{ccccccccc}
  p&0&0&0&-p&0&0&0&0\\
  0&-p&0&-p&0&0&0&0&0\\
  0&0&0&0&0&0&0&0&0
\end{array}\right)
\end{aligned}
\end{eqnarray}
and
\begin{eqnarray}
\begin{aligned}
D=\left(\begin{array}{ccccccccc}
  plmn&0&0&0&-plmn&0&0&0&0\\
  0&-plmn&0&-plmn&0&0&0&0&0\\
  0&0&0&0&0&0&0&0&T
\end{array}\right),
\end{aligned}
\end{eqnarray}
where $D=(D_{j,ik})$, $D_{ijk}=\textrm{Tr}[\rho_{GHZ}(\alpha_i\otimes \beta_j\otimes \gamma_k)]$. The singular values of $D$ are $\sqrt{2}plmn$, $\sqrt{2}plmn$ and $T=\frac{(l^2-1)(m^2n^2+1)}{4}+\frac{(l^2+1)(m^2n^2-1)}{4}p$. $\tilde{\rho}_{GHZ}$ is locally unitary equivalent to $\rho_{GHZ}$. Then we have that $\frac{\sqrt{2}plmn}{F}$, $\frac{\sqrt{2}plmn}{F}$ and $\frac{T}{F}$ are the singular values of the matrix $\frac{\tilde{D}}{F}$, where
\begin{eqnarray}\nonumber
F=\textrm{Tr}[\rho_{GHZ}(\Sigma^2_A\otimes\Sigma^2_B\otimes\Sigma^2_C)]=\frac{(l^2+1)(m^2n^2+1)}{4}+\frac{(l^2-1)(m^2n^2-1)}{4}p.
\end{eqnarray}
The maximal singular value is $\lambda^{'}_{\max}=\frac{\sqrt{2}plmn}{F}$ for given $p$ with $\frac{\sqrt{2}plmn}{F}>\frac{T}{F}$. Then the upper bound of the maximal value of the Mermin operator is $2\sqrt{2}\lambda{'}_{\max}=\frac{4plmn}{F}$. Two singular vectors corresponding to the singular value $\lambda{'}_{\max}$ with degeneracy 2 can be chosen as $\vec{v}_1=(-1,0,0,0,1,0,0,0,0)^T$ and $\vec{v}_2=(0,1,0,1,0,0,0,0,0)^T$, which can be further written as
\begin{eqnarray}
&\vec{v}_1=(1,0,0)^T\otimes(-1,0,0)^T-(0,-1,0)^T\otimes(0,1,0)^T\nonumber,\\
&\vec{v}_2=(1,0,0)^T\otimes(0,1,0)^T+(0,-1,0)^T\otimes(-1,0,0)^T\nonumber.
\end{eqnarray}
By taking $\vec{a}=(1,0,0)^T$, $\vec{a}'=(0,-1,0)^T$, $\vec{c}=(-1,0,0)^T$, $\vec{c}'=(0,1,0)^T$, $\vec{b}$ and $\vec{b}'$ some suitable unit vectors, the upper bound $\frac{4plmn}{F}$ of $\rho_{GHZ}$ is attained. Therefore, the state violates the Mermin inequality if $\frac{4plmn}{F}>2$ under the restriction $\frac{\sqrt{2}plmn}{F}>\frac{T}{F}$. As a result, the standard nonlocality of the state $\rho'_{GHZ}$ can be detected by the Mermin inequality for $0.471428<p\leq 1$. However, the state violates the Mermin inequality if $\frac{1}{2}<p\leq 1$ \cite{QIP2019}. Hence, the hidden standard nonlocality of $\rho_{GHZ}$ is revealed by local filtering operation for $0.471428\leq p\leq 0.5$, see FIG. 1.
\begin{figure}[ht]
\centering
\includegraphics[scale=0.6]{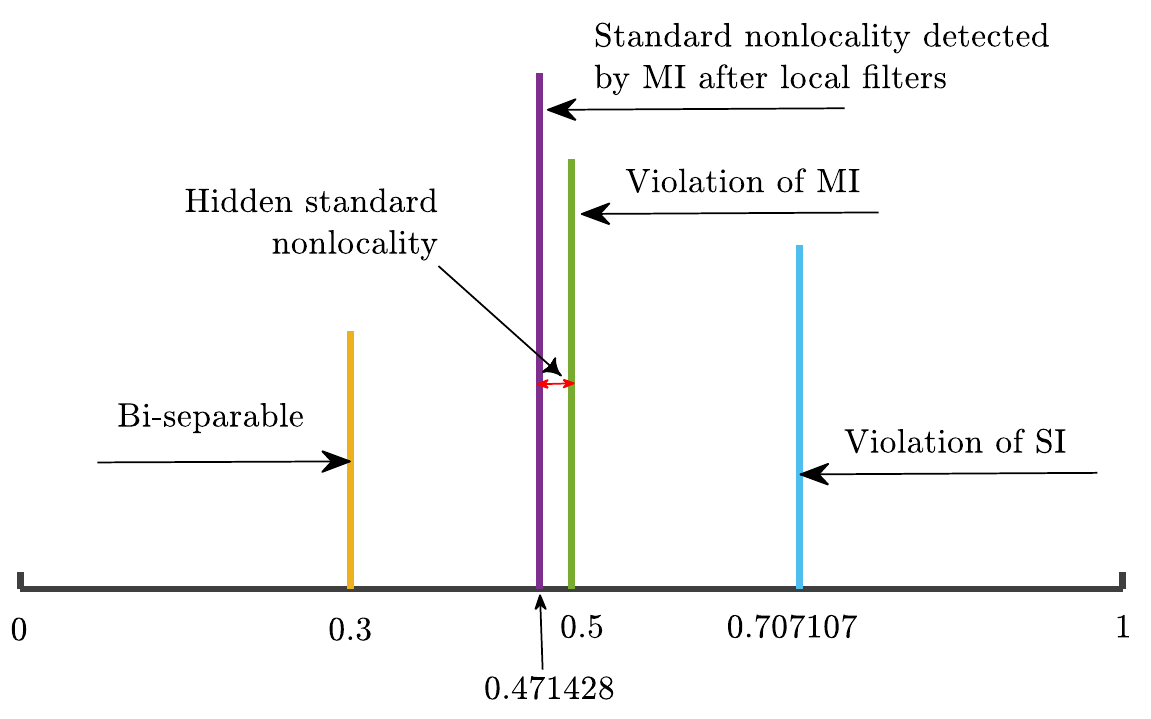}
\caption{The state $\rho_{GHZ}$ violates Mermin inequality (MI) for $0.5<p\leq 1$. The locally filtered state shows standard nonlocality for $0.471428<p\leq 1$. The hidden standard tripartite nonlocality is revealed for $0.471428\leq p\leq 0.5$.}
\end{figure}

\textit{\textbf{Example 2.}} Consider the following state given in \cite{EX22007},
\begin{eqnarray}
\begin{aligned}
\rho =p |\Psi\rangle\langle\Psi|+(1-p)|00\rangle\langle00|\otimes\frac{I_2}{2},
\end{aligned}
\end{eqnarray}
where $0\leq p\leq 1$, $|\Psi\rangle=\cos\frac{\pi}{8}|000\rangle+\sin\frac{\pi}{8}|111\rangle$.
Under local filtering we have
\begin{eqnarray}
\begin{aligned}
D=\left(\begin{array}{ccccccccc}
  \frac{plmn}{\sqrt{2}}&0&0&0&-\frac{plmn}{\sqrt{2}}&0&0&0&0\\
  0&-\frac{plmn}{\sqrt{2}}&0&-\frac{plmn}{\sqrt{2}}&0&0&0&0&0\\
  0&0&0&0&0&0&0&0&T
\end{array}\right).
\end{aligned}
\end{eqnarray}
The singular values of $D$ are $plmn$, $plmn$ and $T=\frac{l^2m^2(n^2-1)}{2}+\frac{-2+\sqrt{2}+2l^2m^2+\sqrt{2}l^2m^2n^2}{4}p$. Since $\tilde{\rho}$ is locally unitary equivalent to $\rho$, the singular value of the matrix $\frac{\tilde{D}}{F}$ are $\frac{plmn}{F}$, $\frac{plmn}{F}$ and $\frac{T}{F}$, where
\begin{eqnarray}\nonumber
F=\textrm{Tr}[\rho(\Sigma^2_A\otimes\Sigma^2_B\otimes\Sigma^2_C)]=\frac{l^2m^2(n^2+1)}{2}+\frac{2-\sqrt{2}-2l^2m^2+\sqrt{2}l^2m^2n^2}{4}p.
\end{eqnarray}
The maximal singular value is $\lambda'_{\max}=\frac{plmn}{F}$ for given $p$ with $\frac{plmn}{F}>\frac{T}{F}$. Then the upper bound of the maximal value of the Mermin operator is $2\sqrt{2}\lambda'_{\max}=\frac{2\sqrt{2}plmn}{F}$. This bound can be attained by selecting the two singular vectors, corresponding to the singular value $\lambda'_{\max}$ with degeneracy 2, to be $\vec{v}_1=(-1,0,0,0,1,0,0,0,0)^T$ and $\vec{v}_2=(0,1,0,1,0,0,0,0,0)^T$, which can be decomposed to
 \begin{eqnarray}
&\vec{v}_1=(1,0,0)^T\otimes(-1,0,0)^T-(0,-1,0)^T\otimes(0,1,0)^T\nonumber,\\
&\vec{v}_2=(1,0,0)^T\otimes(0,1,0)^T+(0,-1,0)^T\otimes(-1,0,0)^T\nonumber.
\end{eqnarray}
Let $\vec{a}=(1,0,0)^T$, $\vec{a}'=(0,-1,0)^T$, $\vec{c}=(-1,0,0)^T$, $\vec{c}'=(0,1,0)^T$. Together with some suitable unit vectors $\vec{b}$ and $\vec{b}'$, the upper bound $\frac{2\sqrt{2}plmn}{F}$ is attained. Therefore, the state violates the Mermin inequality if $\frac{2\sqrt{2}plmn}{F}>2$ under the restriction $\frac{plmn}{F}>\frac{T}{F}$, namely, the state $\rho$ violates the Mermin inequality if $0.318675 < p\leq 1$, for which the standard nonlocality of the state $\rho'$ is detected. The maximal violation of the Mermin inequality is shown in FIG. 2.
\begin{figure}[ht]
\centering
\includegraphics[scale=0.6]{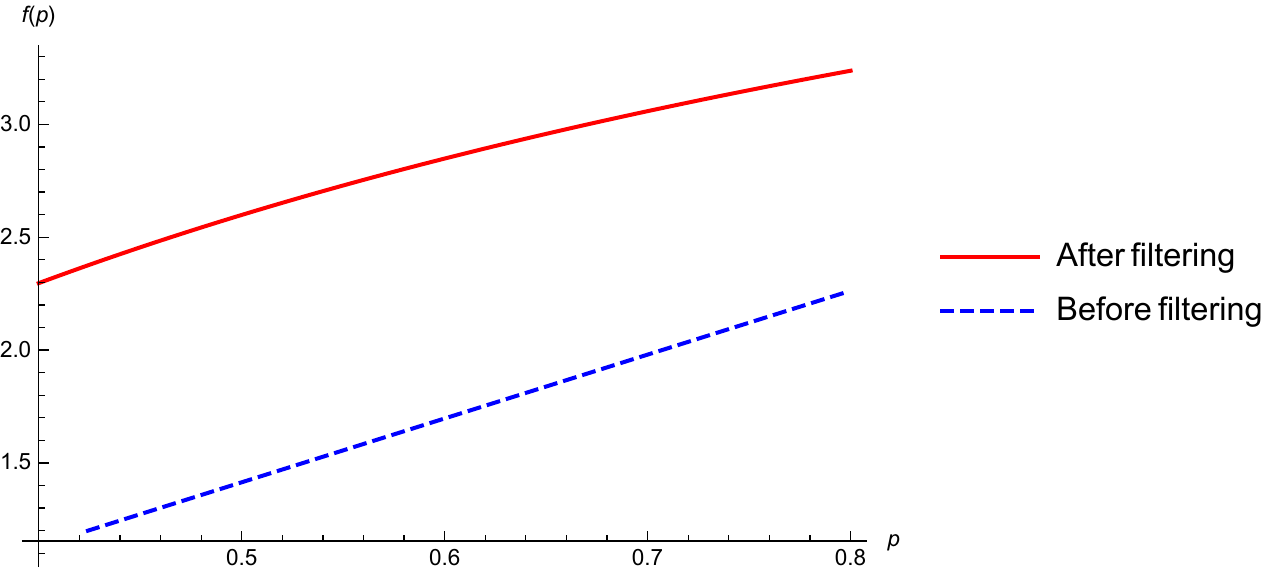}
\caption{$f(p)$ is the maximal value of $\mathcal{Q}(M)$. The red line represents the maximal violation of the locally filtered state $\rho'$. The blue line represents the maximal violation value of the initial state $\rho$.}
\end{figure}

Based on the protocol introduced in \cite{QIP2019}, $\rho$ is standard tripartite nonlocal for $0.707107<p\leq 1$. Therefore, the state $\rho$ shows hidden standard tripartite nonlocality for $0.318675\leq p\leq 0.707107$, see FIG. 3.
\begin{figure}[ht]
\centering
\includegraphics[scale=0.5]{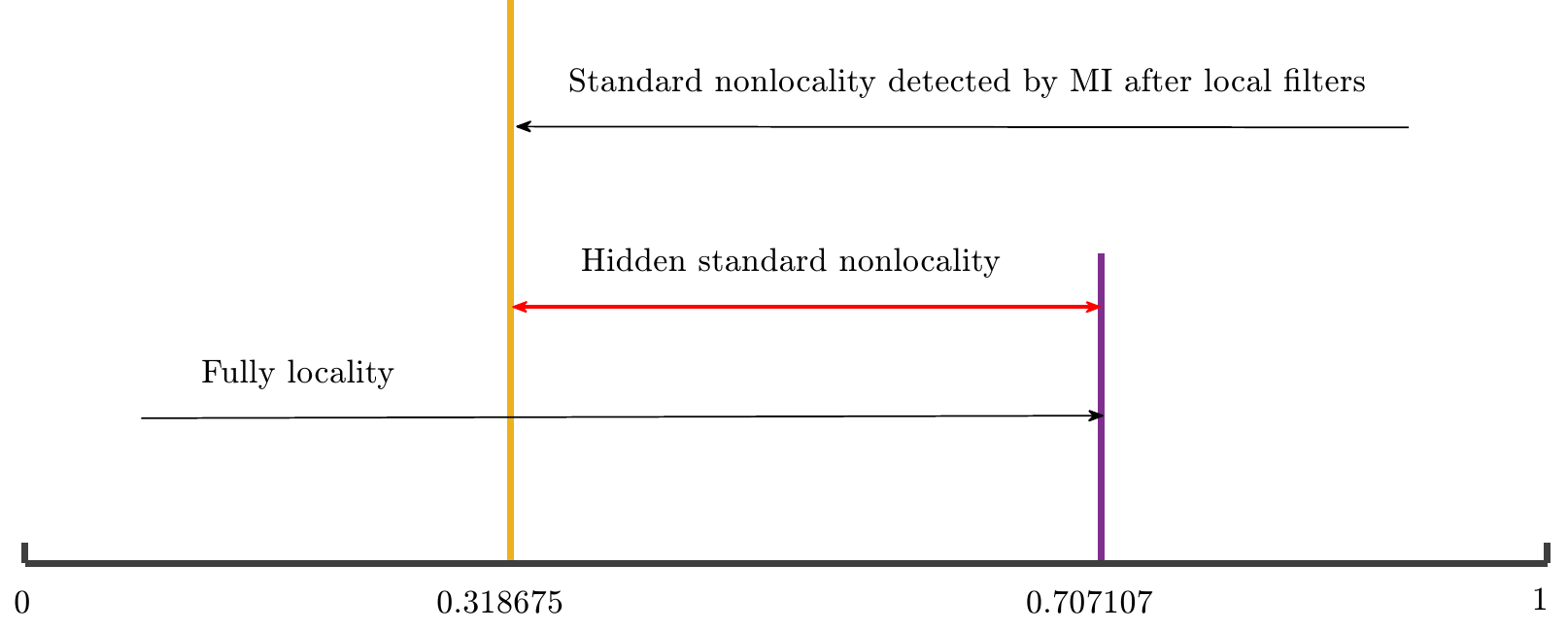}
\caption{The state $\rho$ does not violate SI for $0\leq p\leq 1$. It violates the Mermin inequality for $0.707107< p\leq 1$. After locally filtering $\rho'$ is standard nonlocal for $0.318675< p\leq 1$. The hidden standard nonlocality is revealed for $0.318675\leq p\leq 0.707107$.}
\end{figure}

\textit{\textbf{Example 3.}} The interaction between a quantum system and its environment may reduce the entanglement and nonlocality of the system. The GHZ state is a genuine tripartite nonlocal state as violates the Svetlichny inequality. Let us consider that the GHZ state goes through the amplitude damping(AD) noise channel which maps a qubit state $\rho$ to $\mathcal{E}_{AD}(\rho)=E_0\rho E^{\dag}_0+E_1\rho E^{\dag}_1$, where the Kraus operators are given by \cite{Nielsen},
 $$E_0=\left(\begin{array}{cc}
  1&0\\
  0&\sqrt{1-\gamma}
\end{array}\right),~~~
E_1=\left(\begin{array}{cc}
  0&\sqrt{\gamma}\\
  0&0
\end{array}\right),$$
$\gamma\in (0,1)$ is the damping rate.

When each qubit of the GHZ state $\rho_G$ undergoes the amplitude damping noise channel, one gets
\begin{eqnarray}\label{GHZ1}
\begin{aligned}
\rho_G^{AD}&\equiv\mathcal{E}_{AD}(\rho_G)\\
&=\frac{1}{2}\Big((1+\gamma^3)|000\rangle\langle000|+(1-\gamma )^3|111\rangle\langle111|\\
&\quad+(1-\gamma )^{3/2}(|000\rangle\langle111|+|111\rangle\langle000|)\\
&\quad+(1-\gamma )\gamma ^2(|001\rangle\langle001|+|010\rangle\langle010|+|100\rangle\langle100|)\\
&\quad+(1-\gamma )^2 \gamma(|011\rangle\langle011|+|101\rangle\langle101|+|110\rangle\langle110|)\Big).
\end{aligned}
\end{eqnarray}
The corresponding correlation matrix is
\begin{eqnarray}
\begin{aligned}
C=\left(\begin{array}{ccccccccc}
  s&0&0&0&-s&0&0&0&0\\
  0&-s&0&-s&0&0&0&0&0\\
  0&0&0&0&0&0&0&0&t
\end{array}\right),
\end{aligned}
\end{eqnarray}
where $s=(1-\gamma)^{\frac{3}{2}}$ and $t=\gamma(3-6\gamma+4\gamma^2)$. The singular values of $C$ are $\sqrt{2 (1-\gamma)^{\frac{3}{2}}}$, $\sqrt{2 (1-\gamma)^{\frac{3}{2}}}$ and $\gamma(4\gamma^2-6\gamma+3)$. We choose the two nine-dimension vectors to be $\vec{v}_1=(-1,0,0,0,1,0,0,0,0)^T$ and $\vec{v}_2=(0,1,0,1,0,0,0,0,0)^T$, which can be decomposed into
\begin{eqnarray}
&\vec{v}_1=(1,0,0)^T\otimes(-1,0,0)^T-(0,-1,0)^T\otimes(0,1,0)^T\nonumber,\\
&\vec{v}_2=(1,0,0)^T\otimes(0,1,0)^T+(0,-1,0)^T\otimes(-1,0,0)^T\nonumber.
\end{eqnarray}
Let $\vec{a}=(1,0,0)^T$, $\vec{a}'=(0,-1,0)^T$, $\vec{c}=(-1,0,0)^T$, $\vec{c}'=(0,1,0)^T$. Together with suitable unit vectors $\vec{b}$ and $\vec{b}'$, the upper bound of the maximal expectation of the Mermin operator is $4\sqrt{(1-\gamma)^{\frac{3}{2}}}$ based on \cite{Mermin1990}. Hence, the state is standard tripartite nonlocal for $\gamma\in(0,0.370039)$.

Now consider the filtering. The correlation matrix of the filtered state is
\begin{eqnarray}
\begin{aligned}
D=\left(\begin{array}{ccccccccc}
  S&0&0&0&-S&0&0&0&0\\
  0&-S&0&-S&0&0&0&0&0\\
  0&0&0&0&0&0&0&0&T
\end{array}\right),
\end{aligned}
\end{eqnarray}
where $S=lmn(1-\gamma)^{\frac{3}{2}}$. The singular values of $D$ are $\sqrt{2(1-p)^3l^2m^2n^2}$, $\sqrt{2(1-p)^3l^2m^2n^2}$ and
\begin{eqnarray*}
T&=&\frac{1}{2}(-1+l^2m^2n^2+\gamma^3(l^2+1)(m^2+1)(n^2+1)+\gamma(3+l^2+m^2+n^2)\\
&&-\gamma^2(3+2n^2+m^2(2+n^2)+l^2(2+m^2+n^2))).
\end{eqnarray*}
As $\tilde{\rho}^{AD}_G$ is locally unitary equivalent to $\rho^{AD}_G$, the singular values of the matrix $\frac{\tilde{D}}{F}$ are $\frac{\sqrt{2(1-p)^3l^2m^2n^2}}{F}$, $\frac{\sqrt{2(1-p)^3l^2m^2n^2}}{F}$ and $\frac{T}{F}$, where
\begin{eqnarray*}
F&=&\frac{1}{2}(1+l^2m^2n^2+\gamma^3(l^2-1)(m^2-1)(n^2-1)\\
&&+\gamma(-3+l^2+m^2+n^2)+\gamma^2(3-2n^2+m^2(-2+n^2)+l^2(-2+m^2+n^2))).
\end{eqnarray*}
The maximal violation value of the Mermin operator is $\frac{4\sqrt{(1-\gamma)^3l^2m^2n^2}}{F}$, with the restriction $\frac{\sqrt{2(1-p)^3l^2m^2n^2}}{F}\geq\frac{T}{F}$. Therefore, the filtered state violates the Mermin inequality for $\gamma\in(0, 0.394752)$. That is to say, for the state with $0.370069\leq\gamma\leq0.394752$, the state $\rho^{AD}_G$ is not a standard nonlocal state. FIG. 4 shows that for $\gamma\in(0, 0.394752)$, the filtered state violates the Mermin inequality, i.e., it is a standard tripartite nonlocal state.
\begin{figure}[ht]
\centering
\includegraphics[scale=0.6]{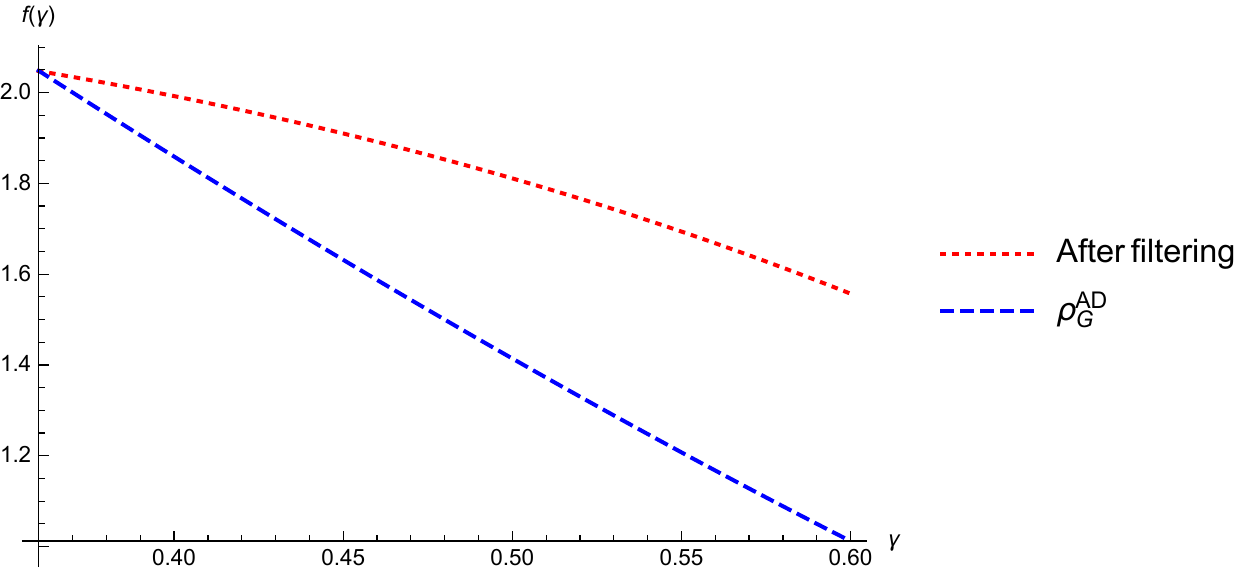}
\caption{$f(\gamma)$ is the maximal value of $\mathcal{Q}(M)$. The red line represents the maximal violation value of the state after local filtering. The filtered state is standard nonlocal for $0\leq \gamma<0.394752$.}
\end{figure}

\section{Conclusion}
In summary, we have investigated the maximal violation of the Mermin inequality under local filtering for any 3-qubit states. We have presented a tight upper bound for the maximal expectation value of the Mermin operator after local filtering. Furthermore, for the 3-qubit GHZ state, the standard tripartite nonlocal be revealed for $0.471428\leq p\leq 0.5$ by local filtering. Similarly, although the amplitude damping GHZ state is fully local for $0.370069\leq\gamma\leq0.394752$, the filtered one is standard nonlocal.
The local filtering process may reveal certain hidden quantum correlations including nonlocality \cite{filterCHSH,FOP} and steerability \cite{FilterSteer}. In order to improve the efficiency of quantum information process, a number of scheme have been put froward \cite{purify,ErrorCorrection,EntanglementConcentration,QuantumRepeaters}. The filter operations can also be used to improve the fidelity between quantum states and efficiency of information processing with noisy entangled state \cite{FilterNoise}. Our approach presented in this article can also be used to deal with other Bell-type inequalities for tripartite or multipartite systems.

\begin{acknowledgements} This work is supported by NSFC (Grant Nos. 12075159, 12171044), Beijing Natural Science Foundation (Z190005), Academy for Multidisciplinary Studies, Capital Normal University, the Academician Innovation Platform of Hainan Province, and Shenzhen Institute for Quantum Science and Engineering, Southern University of Science and Technology (SIQSE202001).
\end{acknowledgements}

\footnotesize{\noindent\textbf{Data Availability Statements} All data generated or analysed during this study are available from the corresponding author on reasonable request.}


\begin{thebibliography}{99}
\bibitem{EPRnonlocal} Einstein, A., Podolsky, B., and Rosen, N.: Can Quantum-Mechanical Description of Physical Reality Be Considered Complete? \href{https://doi.org/10.1103/PhysRev.47.777}{Phys. Rev. \textbf{47}, 777 (1935)}.
\bibitem{BellNonlocal} Brunner, N., Cavalcanti, D., Pironio, S., Scarani, V., and Wehner, S.: Bell nonlocality, \href{https://doi.org/10.1103/RevModPhys.86.419}{Rev. Mod. Phys. \textbf{86}, 419 (2014)}.
\bibitem{Wiseman} Wiseman, H. M., Jones, S. J., and Doherty, A. C.: Steering, Entanglement, Nonlocality, and the Einstein-Podolsky-Rosen Paradox,  \href{https://doi.org/10.1103/PhysRevLett.98.140402}{Phys. Rev. Lett. \textbf{98}, 140402 (2007)}.
\bibitem{StrongNonlocal} Halder, S., Banik, M., Agrawal, S., and Bandyopadhyay, S.: Strong Quantum Nonlocality without Entanglement,   \href{https://doi.org/10.1103/PhysRevLett.122.040403}{Phys. Rev. Lett. \textbf{120}, 040403 (2019)}.
\bibitem{HHHH2009entangle} Horodecki, R., Horodecki, P., Horodecki, M., and Horodecki, K.: Quantum entanglement, \href{https://doi.org/10.1103/RevModPhys.81.865}{Rev. Mod. Phys. \textbf{81}, 865 (2009)}.
\bibitem{LS2019} Levine, Y., Sharir, O., Cohen, N., and Shashua, A.: Quantum Entanglement in Deep Learning Architectures, \href{https://doi.org/10.1103/PhysRevLett.122.065301}{Phys. Rev. Lett. \textbf{122}, 065301 (2019)}.
\bibitem{CG2005} Cerf, V. J., Gisin, N., Massar, S., and Popescu, S.: Simulating Maximal Quantum Entanglement without Communication,  \href{https://doi.org/10.1103/PhysRevLett.94.220403}{Phys. Rev. Lett. \textbf{94}, 220403 (2005)}.
\bibitem{steer2020} Uola, R., Costa, A. C. S., Nguyen, H. C., and G\"{u}hne, O.: Quantum steering,  \href{https://doi.org/10.1103/RevModPhys.92.015001}{Rev. Mod. Phys. \textbf{92}, 015001 (2020)}.
\bibitem{Bell1966} BELL, J. S.: On the Problem of Hidden Variables in Quantum Mechanics, \href{https://doi.org/10.1103/RevModPhys.38.447}{Rev. Mod. Phys. \textbf{38}, 447 (1966)}.
\bibitem{Liang2020} Liang, J. M., Shen, S. Q., Li, M., and Li, L.: Variational quantum algorithms for dimensionality reduction and classification, \href{https://doi.org/10.1103/PhysRevA.101.032323}{Phys. Rev. A \textbf{101}, 032323 (2020)}.
\bibitem{Liang2022} Liang, J. M., Wei, S. J., and Fei, S. M.: Quantum gradient descent algorithms for nonequilibrium steady states and linear algebraic systems, \href{https://doi.org/10.1007/s11433-021-1844-7}{Sci. China Phys. Mech. Astron. \textbf{65}, 250313 (2022)}.
\bibitem{AK1991} Ekert, A. K.: Quantum Cryptography based on Bell's theorem, \href{https://doi.org/10.1103/PhysRevLett.67.661}{Phys. Rev. Lett. \textbf{67}, 661 (1991)}.
\bibitem{Nature2010} Pironio, S., Ac\'{i}n, A., Massar, S., Boyer de la Giroday, A., Matsukevich, D. N., Maunz, P., Olmschenk, S., Hayes, D., Luo, L., Manning, T. A., and Monroe, C.: Random numbers certified by Bell's theorem, \href{https://doi.org/10.1038/nature09008}{Nature \textbf{464}, 1021 (2010)}.
\bibitem{CHSH1969} Clauser, J. F., Horne, M. A., Shimony, A., and Holt, R. A.: Proposed Experiment to Test Local Hidden-Variable Theories,  \href{https://doi.org/10.1103/PhysRevLett.23.880}{Phys. Rev. Lett. \textbf{23}, 880 (1969)}.
\bibitem{CHSHviolate} Verstraete, F. and Wolf, M. M.: Entanglement versus Bell Violations and Their Behavior under Local Filtering Operations, \href{https://doi.org/10.1103/PhysRevLett.89.170401}{Phys. Rev. Lett. \textbf{89}, 170401 (2001)}.
\bibitem{filterCHSH} Hirsch, F., Quintino, M. T., Bowles, J., and Brunner, N.: Genuine Hidden Quantum Nonlocality,  \href{https://doi.org/10.1103/PhysRevLett.111.160402}{Phys. Rev. Lett. \textbf{111}, 160402 (2013)}.
\bibitem{FilterSteer} Pramanik, T., Cho, Y. W., Han, S. W., Lee, S. Y., Kim, Y. S., and Moon, S.: Revealing hidden quantum steerability using local filtering operations, \href{https://doi.org/10.1103/PhysRevA.99.030101}{Phys. Rev. A \textbf{99}, 030101(R) (2019)}.
\bibitem{Mermin1990} Mermin, N. D.: Extreme Quantum Entanglement in a Superposition of Macroscopically Distinct States, \href{https://doi.org/10.1103/PhysRevLett.65.1838}{Phys. Rev. Lett. \textbf{65}, 1838 (1990)}.
\bibitem{QIP2019} Siddiqui,  M. A. and Sazim, S.: Tight upper bound for the maximal expectation value of the Mermin operators, \href{https://doi.org/10.1007/s11128-019-2246-1}{Quantum. Inf. Process. \textbf{18}, 131 (2019)}.
\bibitem{FOP} Sun, L. Y., Xu, L., Wang, J., Li, M., Shen, S. Q., Li, L., and Fei, S. M.: Tight upper bound on the quantum value of Svetlichny operators under local filtering and hidden genuine nonlocality, \href{https://doi.org/10.1007/s11467-020-1015-z}{Front. Phys. \textbf{16}, 31501 (2021)}.
\bibitem{SO1995} Schlienz, J. and Mahler, G.: Description of entanglement, \href{https://doi.org/10.1103/PhysRevA.52.4396}{Phys. Rev. A \textbf{52}, 4396 (1995)}.
\bibitem{EX12015} Augusiak, R., Demianowicz, M., Tura, J., and Ac\'{i}n, A.: Entanglement and Nonlocality are Inequivalent for Any Number of Parties,  \href{https://doi.org/10.1103/PhysRevLett.115.030404}{Phys. Rev. Lett. \textbf{115}, 030404 (2015)}.
\bibitem{GN2013} Bancal, J. D., Barrett, J., Gisin, N., and Pironio, S.: Definitions of multipartite nonlocality, \href{https://doi.org/10.1103/PhysRevA.88.014102}{Phys. Rev. A \textbf{88}, 014102 (2013)}.
\bibitem{LM2017} Li, M., Shen, S. Q., Jing, N. H., Fei, S. M., and Li-Jost, X. Q.: Tight upper bound for the maximal quantum value of the Svetlichny operators, \href{https://doi.org/10.1103/PhysRevA.96.042323}{Phys. Rev. A \textbf{96}, 042323 (2017)}.
\bibitem{Svetlichny1987} Svetlichny, G.: Distinguishing three-body from two-body nonseparability by a Bell-type inequality,  \href{https://doi.org/10.1103/PhysRevD.35.3066}{Phys. Rev. D \textbf{35}, 3066 (1987)}.
\bibitem{EX22007} Almeida, M. L., Pironio, S., Barrett, J., T\'{o}th, G., and Ac\'{i}n, A.: Noise Robustness of the Nonlocality of Entangled Quantum States, \href{https://doi.org/10.1103/PhysRevLett.99.040403}{Phys. Rev. Lett. \textbf{99}, 040403 (2007)}.
\bibitem{Nielsen} Nielsen, M. A. and Chuang, I. L.: \emph{Quantum Computation and Quantum Information}, (Cambridge University Press, Cambridge, England, 2000).
\bibitem{purify} Bennett, C. H., Brassard, G., Popescu, S., Schumacher, B., Smolin, J. A., and Wootters, W. K.: Purification of Noisy Entanglement and Faithful Teleportation via Noisy Channels, \href{https://doi.org/10.1103/PhysRevLett.76.722}{Phys. Rev. Lett. \textbf{76}, 722 (1997)}.
\bibitem{ErrorCorrection} Bennett, C. H., DiVincenzo, D. P., Smolin, J. A., and Wootters, W. K.: Mixed-state entanglement and quantum error correction, \href{https://doi.org/10.1103/PhysRevA.54.3824}{Phys. Rev. A \textbf{54}, 3824 (1996)}.
\bibitem{EntanglementConcentration} Zhao, Z., Pan, J. W., and Zhan, M. S.: Practical scheme for entanglement concentration, \href{https://doi.org/10.1103/PhysRevA.64.014301}{Phys. Rev. A \textbf{64}, 014301 (2001)}.
\bibitem{QuantumRepeaters} Briegel, H. J., D\"{u}r, W., Cirac, J. I., and Zoller, P.: Quantum Repeaters: The Role of Imperfect Local Operations in Quantum Communication, \href{https://doi.org/10.1103/PhysRevLett.81.5932}{Phys. Rev. Lett. \textbf{81}, 5932 (1998)}.
\bibitem{FilterNoise} Huang, Y. S., Xing, H. B., Yang, M., Yang, Q., Song, W., and Cao, Z. L.: Distillation of multipartite entanglement by local filtering operations, \href{https://doi.org/10.1103/PhysRevA.89.062320}{Phys. Rev. A \textbf{89}, 062320 (2014)}.
\end{thebibliography}
\end{document}